\documentclass{ifcolog}

\usepackage{makeidx}
\usepackage{amsthm}
\usepackage{amssymb}
\usepackage{algorithm}
\usepackage{algpseudocode}
\usepackage{float}

\newcommand{\Pdt}{P_{\mathrm{D}}^{*}}
\newcommand{\Plp}{P_{\mathrm{D}}^{\mathrm{LP}}}
\newcommand{\Sd}{S_{\mathrm{D}}}
\newcommand{\rr}{\bar{r}}

\newcommand{\sbank}{s_{\textrm{BANK}}}
\newcommand{\slp}{s_{\mathrm{LP}}}
\newcommand{\soracle}{s_{\textrm{O}}}
\newcommand{\sfree}{s_{\textrm{FREE}}}
\newcommand{\sarb}{s_{\textrm{ARB}}}
\newcommand{\sib}{s_{\textrm{IB}}}
\newcommand{\sburn}{s_{\textrm{BURN}}}
\newcommand{\ubank}{U_{BANK}}
\newcommand{\ulp}{U_{LP}}

\newcommand{\tautx}{\tau_{\textrm{TX}}}

\newcommand{\epsilonfree}{\epsilon_{\textrm{FREE}}}
\newcommand{\epsilonib}{\epsilon_{\textrm{IB}}}
\newcommand{\epsilondexyzero}{\epsilon_{\textrm{ED}}^{0}}
\newcommand{\epsilondexyone}{\epsilon_{\textrm{ED}}^{1}}
\newcommand{\epsilonreleasedexy}{\epsilon_{\textrm{RD}}}

\newcommand{\dx}{Dexy}
\newcommand{\dxg}{DexyGold}
\newcommand{\dxusd}{USE}

\newtheorem{definition}{Definition}
\newtheorem{theorem}{Theorem}
\newtheorem{corollary}{Corollary}

\title{Dexy:\protect \\ A Simple Stablecoin Design Based on an Algorithmic Central Bank}

\titlerunning{Dexy}

\addauthor[kushti@protonmail.ch]
  {Alexander Chepurnoy}
  {Ergo Platform}
\addauthor[]
  {Amitabh Saxena}
  {Ergo Platform}
\addauthor[ldgaetano@protonmail.com]
  {Luca D'Angelo}
  {The Stable Order}

\authorrunning{A. Chepurnoy and A. Saxena and L. D'Angelo}

\titlethanks{
We would like to thank Ile for his inspiring forum posts, 
Bruno for discussions and the extract-to-future bank action idea, and 4EYES Labs for 
the audit of the ErgoScript contracts.
}

\begin{document}
\maketitle

\begin{abstract}
We consider a new algorithmic stablecoin called Dexy consisting of three components: 
a central bank, a secondary market in the form of a customized CP-AMM liquidity pool, and a trusted oracle. 
The price of the stablecoin in the liquidity pool is stabilized, relative to the oracle price, via central bank interventions. 
Users mint Dexy stablecoins from the bank using basecoins, the base cryptocurrency asset.
\end{abstract}

\section{Introduction}
\label{sec:introduction}
Due to the inherent volatility of cryptocurrency prices, algorithmic stablecoins naturally arose as their extension 
to address this problem. The volatility is resolved by pegging the stablecoin to an asset whose price is 
considered to be stable over long periods of time, e.g. gold. 

Having an asset with a stable value can be useful in many scenarios:
\begin{itemize}
\item Securing fundraising: a project can be sure that funds collected during fundraising will have stable value in the mid-to-long term.
\item Doing business with predictable results: a shop can be sure that funds collected from sales will have roughly the same value when the shop is ordering goods from warehouses, otherwise, the shop may go bankrupt if its margins are too low to cover exchange rate fluctuations. 
\item Shorting: when cryptocurrency prices are high, it is desirable for investors to rebalance their portfolio by increasing exposure to fiat currencies or real-world commodities. However, as fiat currencies and centralized exchanges impose significant risks, it would be better to buy fiat and commodity substitutes in the form of stablecoins on decentralized exchanges.
\item Lending and other decentralized finance applications: stability of collateral value is critical for many applications built on top of these services.
\end{itemize}

Centralized stablecoins, such as USDT and USDC, use a trusted party to ensure that the peg is maintained. For algorithmic stablecoins, 
the peg can be maintained in different ways depending on the implementation, e.g. rebasing of the total supply \cite{Ampleforth}, 
or via imitating a trusted party that holds reserves and then performs market interventions for rebalancing the peg. 
Imitating the trusted party is usually done by allowing anyone on the blockchain to create over-collateralized financial instruments, 
such as collateralized debt positions \cite{DAI}, zero-coupon bonds (as was done in the failed Yield protocol), 
reserve assets \cite{DjedEprint,Gluon}, or by issuing stabilizing financial instruments in case of a depeg \cite{Neutrino}.

In this work, we present \dx{}, a stablecoin protocol where an algorithmic central bank performs interventions when there is a depeg.
The bank stabilizes the stablecoin value in the reference market, a liquidity pool, by injecting basecoins from its reserves when the 
stablecoin price is below the peg. In extreme cases, when bank reserves are depleted and the stablecoin price is still below the peg, 
the liquidity pool price is restored by extracting stablecoins from it; these will be injected back again at a future time, 
when the stablecoin price is above the peg by a certain threshold. 

The rest of the paper is organized as follows. In section \ref{sec:design}, the \dx{} design in general is explained, while 
section \ref{sec:protocol} defines the protocol specification in detail.
In section \ref{sec:theorems} we consider different theorems related to the stability and design of the protocol.
Finally, section \ref{sec:implementations} describes two \dx{} implementations: the Gold-pegged DexyGold stablecoin and the USD-pegged USE stablecoin.

\section{\dx{} Design}
\label{sec:design}

The \dx{} protocol functions via an algorithmic central bank that holds reserves in basecoins and mints stablecoins. 
A corresponding liquidity pool exists where stablecoins and basecoins can be traded. 
The peg price used by the bank is provided by an oracle, which is assumed to be trusted and reliable. 
Users can interact with both the bank and the liquidity pool, but certain rules and restrictions are applied to control 
the price stability. Central bank interventions are determined by the ratio of the liquidity pool price to the oracle price, $\rr$, 
which we also call the reference reserve ratio. 

\subsection{The Liquidity Pool}
\label{subsec:pool}

The reference market is implemented as an automated market maker (e.g. \cite{xu2021sok}). 
For simplicity, a modified constant-product automated market maker (CP-AMM) is used, 
similar to ErgoDex \cite{ErgoDex} and Uniswap \cite{Uniswap}. The liquidity pool holds a pair of 
basecoins and stablecoins. Liquidity providers deposit a pair of basecoins 
and stablecoins and receive LP tokens in proportion to their contribution. 
Each trade is required to leave the product of the two reserves non-decreasing: 
a swap fee is charged on the traded-in amount and retained in the pool, accruing 
for all LP token holders pro-rata. LP tokens can later be redeemed to withdraw a 
proportional share of the liquidity pool's current reserves, 
minus a redemption fee used to disincentivize frequent withdrawals.
Liquidity provider withdrawals are halted entirely whenever the pool price is 
below the oracle price by a certain threshold.

\subsection{The Algorithmic Central Bank}
\label{subsec:bank}

By depositing basecoins into the bank, users can mint stablecoins. 
The price provided by the bank will always equal the oracle price. 
A fee is charged so the bank reserves can increase faster, along with a fee for the oracle usage.
Under certain conditions, the bank can intervene into the liquidity pool to control its price relative to the oracle price.

\subsection{Trust Assumptions}
\label{subsec:kya}

In general, for any DeFi protocol involving user funds, it is important for trust assumptions to be explicitly stated so 
that users can make informed decisions to the best of their ability.

The \dx{} protocol has two trust assumptions:

\begin{itemize}
  \item{Oracle:} The protocol assumes the target price oracle is reliable. This trust assumption is unavoidable, however, 
  it can be relaxed by using a federation of oracles that produces an average price to be delivered to the central bank, 
  instead of relying on a single entity. In any case, it is important that oracles are rewarded from protocol usage, 
  otherwise, they will not be incentivized to keep providing valid data, see Section \ref{subsec:oracle} for details.

  \item{No governance:} Unlike most stablecoin protocols, we do not consider governance, as it can be used to 
  adversarially attack the protocol. We suppose that real-world instantiations that do have it should clearly state their 
  governance-related trust assumptions.
\end{itemize}

\subsection{Actions and Dynamics}
\label{subsec:actions}

The following table summarizes the different actions in the protocol, the $\rr$ range when the action happens, and the restored $\rr'$ range after the action happens.

\begin{table}[h]
\begin{tabular}{|l|l|l|}
\hline
\multicolumn{1}{|c|}{\textbf{Action}} & \multicolumn{1}{c|}{\textbf{Range}}                                     & \multicolumn{1}{c|}{\textbf{Restore Range}}                                                           \\ \hline
FreeMint                              & $1- \epsilon_{\textrm{FREE}} \le \rr \le 1+ \epsilon_{\textrm{FREE}}$   & \multicolumn{1}{c|}{-}                                                                                \\ \hline
ArbitrageMint                         & $\rr > 1 + \epsilon_{\textrm{FREE}}$                                    & \multicolumn{1}{c|}{-}                                                                                \\ \hline
InjectBasecoins                       & $\rr < 1 - \epsilon_{\textrm{FREE}}$                                    & $1 - \epsilon_{\textrm{FREE}} \le \rr' \le 1-\epsilon_{\textrm{IB}}$                                  \\ \hline
ExtractDexy                           & $\rr \le 1 - \epsilon_{\textrm{ED}}^{0} < 1 - \epsilon_{\textrm{FREE}}$ & $1 - \epsilon_{\textrm{FREE}} - \epsilon_{\textrm{ED}}^{1} \le \rr' \le 1-\epsilon_{\textrm{FREE}}$   \\ \hline
ReleaseDexy                           & $1 + \epsilon_{\textrm{RD}} < \rr <  1 + \epsilon_{\textrm{FREE}}$      & $1 + \epsilon_{\textrm{RD}} \le \rr' \le 1 + \epsilonfree$                                                                 \\ \hline
StoppedWithdrawals                    & $\rr < 1 - \epsilon_{\textrm{FREE}}$                                    & \multicolumn{1}{c|}{-}                                                                  \\ \hline
\end{tabular}
\end{table}

We now give a high-level description of what happens during each action. For more technical details, refer to section \ref{sec:protocol}.

\subsubsection{FreeMint}
  \begin{itemize}
    \item $\rr$ is bounded within the range $1- \epsilon_{\textrm{FREE}} \le \rr \le 1+ \epsilon_{\textrm{FREE}}$.
    \item If bounded within the range, then minting can occur for a fixed period of time $T_{\textrm{FREE}}$.
    \item Users can mint from the bank up to a maximum amount determined by the amount of Dexy stablecoins in the liquidity pool at the beginning of the period.
    \item The bank charges a fee so the bank reserves can increase faster.
    \item \textbf{Consequence:} Increase bank reserves and stablecoin supply when the pool price is considered stable.
  \end{itemize}
  
\subsubsection{ArbitrageMint}
  \begin{itemize}
    \item $\rr$ is greater than $1 + \epsilon_{\textrm{FREE}}$.
    \item If $\rr$ is above the threshold for an amount of time greater than $\tau_{\textrm{ARB}}$,
          then minting can occur for a fixed period of time $T_{\textrm{ARB}}$.
    \item The liquidity pool and oracle price difference, at the beginning of the period, determines the max amount of dexy stablecoins to mint $\delta_{d}$.
    \item The bank charges a fee so the bank reserves can increase faster.
    \item \textbf{Consequence:} Purchase stablecoins from the bank at a discount and sell them to the liquidity pool at premium, 
    increasing the liquidity pool stablecoin supply.
  \end{itemize}

\subsubsection{InjectBasecoins}
  \begin{itemize}
    \item $\rr$ is less than $1 - \epsilon_{\textrm{FREE}}$.
    \item $\rr$ is less than the threshold for an amount of time greater than $\tau_{\textrm{IB}}$, upon which the period begins.
    \item An amount of time greater than $T_{\textrm{IB}}$ has elapsed before the first intervention, i.e. since the period began, or since the previous intervention. 
    \item The bank still has basecoin reserves.
    \item The bank will inject basecoins from its reserves into the liquidity pool 
          until the new $\rr$ is within the range $1 - \epsilon_{\textrm{FREE}} \le \rr' \le 1-\epsilon_{\textrm{IB}}$.
    \item The amount of basecoins injected is a fixed percentage $\alpha_{\textrm{IB}}$ of the bank's total basecoin reserve.
    \item \textbf{Consequence:} The bank's basecoin reserve decreases and the liquidity pool price recovers.
  \end{itemize}

\subsubsection{ExtractDexy}
  \begin{itemize}
    \item $\rr$ is less than $1 - \epsilon_{\textrm{ED}}^{0} < 1 - \epsilon_{\textrm{FREE}}$.
    \item $\rr$ is less than the threshold for an amount of time greater than $\tau_{\textrm{ED}} >> \tau_{\textrm{IB}}$, upon which the period begins.
    \item An amount of time greater than $T_{\textrm{BURN}}$ has elapsed before the first extraction, i.e. since the period began, or since the previous extraction. 
    \item The bank has no more basecoin reserves.
    \item The bank will extract stablecoins from the liquidity pool until the new $\rr$ is in 
          the range $1 - \epsilon_{\textrm{FREE}} - \epsilon_{\textrm{ED}}^{1} \le \rr' \le 1-\epsilon_{\textrm{FREE}}$, 
          so essentially the stopped-withdrawal state, and then lock them to be released at a future time.
    \item \textbf{Consequence:} The bank removes stablecoins from the liquidity pool and the liquidity pool price recovers.
  \end{itemize}

\newpage 

\subsubsection{ReleaseDexy}
  \begin{itemize}
    \item $\rr$ is within the range $1 + \epsilon_{\textrm{RD}} < \rr <  1 + \epsilon_{\textrm{FREE}}$.
    \item $\rr$ is within the range for an amount of time greater than $\tau_{\textrm{RD}}$, upon which the period begins.
    \item An amount of time greater than $T_{\textrm{BURN}}$ has elapsed before the first release, i.e. since the period began, or since the previous release. 
    \item The bank will release the locked stablecoins back into the liquidity pool until the new $\rr$ is greater than or equal to $1 + \epsilon_{\textrm{RD}}$
    \item \textbf{Consequence:} The bank releases stablecoins back into the liquidity pool such that it is prioritized over arbitrage minting.
  \end{itemize}

\subsubsection{StoppedWithdrawals}
  \begin{itemize}
    \item $\rr$ is less than $1 - \epsilon_{\textrm{FREE}}$.
    \item Liquidity pair withdrawals are stopped without delay once the $\rr$ intervention threshold is met.
    \item Only swaps between basecoins and stablecoins are allowed.
    \item \textbf{Consequence:} Liquidity providers, i.e. LP token holders, cannot withdraw their liquidity pair from the liquidity pool to prevent further price depegging.
  \end{itemize}

\section{Protocol Specification}
\label{sec:protocol}
We now describe the \dx{} protocol specification, which includes immutable protocol parameters and state variables. 
We note the distinction between an action period, which measures how long before a bank intervention can happen, 
and a corresponding action bank delay, which measures how long before the start of the first bank action period. 
For a particular bank action, there may be multiple periods before a new delay is required.

The central bank parameters are a tuple $C_{\textrm{BANK}}$ of: \\
\resizebox{\textwidth}{!}{$\displaystyle\left( \phi_{\textrm{BANK}}^{\textrm{FREE}}, \phi_{\textrm{BANK}}^{\textrm{ARB}}, \epsilon_{\textrm{FREE}}, \epsilon_{\textrm{IB}}, \epsilon_{\textrm{ED}}^{0}, \epsilon_{\textrm{ED}}^{1}, \epsilon_{\textrm{RD}}, T_{\textrm{FREE}}, T_{\textrm{ARB}}, T_{\textrm{IB}}, T_{\textrm{BURN}}, \tau_{\textrm{ARB}}, \tau_{\textrm{IB}}, \tau_{\textrm{ED}}, \tau_{\textrm{RD}}, \tau_{\textrm{TX}}, \alpha_{\textrm{FREE}}, \alpha_{\textrm{IB}}  \right)$}

\begin{itemize}
  \item $\phi_{\textrm{BANK}}^{\textrm{FREE}} \in (0, 1)$ is the bank fee for the free-mint action.
  \item $\phi_{\textrm{BANK}}^{\textrm{ARB}} \in (0,1)$ is the bank fee for the arbitrage-mint action.
  \item $\epsilon_{\textrm{FREE}} \in (0, 1)$ defines the upper and lower bounds of the free-mint range, representing the percentage that the LP price deviates from the oracle price.
  \item $0 < \epsilon_{\textrm{IB}} < \epsilon_{\textrm{FREE}}$ defines the percentage where the inject-basecoins action will restore the price.
  \item $\epsilon_{\textrm{FREE}} < \epsilon_{\textrm{ED}}^{0} < 1$ defines the percentage where the extract-dexy action will happen.
  \item $0 < \epsilon_{\textrm{ED}}^{1} < \epsilon_{\textrm{FREE}}$ defines the percentage where the extract-dexy action will restore the price.
  \item $\epsilon_{\textrm{RD}} < \epsilon_{\textrm{FREE}}$ defines the percentage where the release-dexy action will happen.
  \item $T_{\textrm{FREE}} \in \mathbb{R}_{>0}$ is the free-mint period.
  \item $T_{\textrm{ARB}} \in \mathbb{R}_{>0}$ is the arbitrage-mint period.
  \item $T_{\textrm{IB}} \in \mathbb{R}_{>0}$ is the inject-basecoins period.
  \item $T_{\textrm{BURN}} \in \mathbb{R}_{>0}$ is the extract-dexy and release-dexy period.
  \item $\tau_{\textrm{ARB}} \in \mathbb{R}_{>0}$ is the arbitrage-mint bank delay.
  \item $\tau_{\textrm{IB}} \in \mathbb{R}_{>0}$ is the inject-basecoins bank delay.
  \item $\tau_{\textrm{ED}} >> \tau_{\textrm{IB}}$ is the extract-dexy bank delay, which must be much larger than the inject-basecoins bank delay, since extraction is an action of last-resort.
  \item $\tau_{\textrm{RD}} \in \mathbb{R}_{>0}$ is the release-dexy bank delay.
  \item $\tau_{\textrm{TX}} \in \mathbb{R}_{>0}$ is the allowed transaction propagation delay.
  \item $\alpha_{\textrm{FREE}} \in (0,1)$ is the percentage of dexy liquidity, in the liquidity pool, that determines the max amount of dexy that can be minted from the bank during the free-mint period.
  \item $\alpha_{\textrm{IB}} \in (0,1)$ is the percentage of the bank's basecoin reserve that determines the max amount of basecoins to inject into the liquidity pool during an inject-basecoin action.
\end{itemize}

The central bank state is determined by a tuple of two values:
$$s_{\textrm{BANK}} = \left(R, S_{D} \right)$$

\begin{itemize}
  \item $\ubank$ is a finite set of users that interact with the bank to mint stablecoins.
  \item $R \in \mathbb{R}_{>0}$ is the bank's basecoin reserves.
  \item $S_D \in \mathbb{R}_{>0}$ is the total supply of minted dexy stablecoins.
\end{itemize}

The liquidity pool has two parameters:
$$C_{LP} = \left(\phi^{REDEEM}_{LP}, \phi^{SWAP}_{LP}\right)$$

\begin{itemize}
  \item $\ulp$ is a finite set of users that interact with the liquidity pool, either as liquidity providers or traders.
  \item $\phi^{REDEEM}_{LP} \in (0,1)$ is the liquidity provider withdrawal fee, which should be high enough to disincentivize frequent withdrawals.
  \item $\phi^{SWAP}_{LP} \in (0, 1)$ is the liquidity pool swap fee, which is used to reward liquidity providers. 
\end{itemize}

The liquidity pool state is determined by a tuple of three values:
$$s_{\textrm{LP}} = \left(B, D, S_{LP} \right)$$

\begin{itemize}
  \item $B \in \mathbb{R}_{>0}$ is the amount of basecoins in the liquidity pool.
  \item $D \in \mathbb{R}_{>0}$ is the amount of dexy stablecoins in the liquidity pool.
  \item $S_{LP} \in \mathbb{R}_{>0}$ is the supply of LP tokens provided to liquidity providers.
\end{itemize}

The oracle parameters are assumed to only be the following:
$$C_{O} = \left(\phi_{O}\right)$$

\begin{itemize}
  \item $\phi_{O} \in (0,1)$ is the oracle fee.
\end{itemize}

The oracle state is determined by a tuple of two values:
$$s_{\textrm{O}} = \left(\Pdt, F \right)$$

\begin{itemize}
  \item $\Pdt \in \mathbb{R}_{>0}$ is the target stablecoin price, in units of $\frac{\textrm{BASECOIN}}{\textrm{PEG ASSET}}$.
  \item $F \in \mathbb{R}_{\ge 0}$ is the accumulated fees of the oracle.
\end{itemize}

For the sake of clarity, we now define separate state variables for the different bank actions, 
but these can also be considered as part of the bank state.

The FreeMint action has the following state:
$$s_{\textrm{FREE}} = \left( t, d \right)$$

\begin{itemize}
  \item $t \in \mathbb{R}_{>0}$ is the time when the free-mint period ends.
  \item $d \in \mathbb{R}_{\ge0}$ is the remaining amount of dexy stablecoins that can be minted during the current free-mint period, determined by $t$.
\end{itemize}

The ArbMint action has the following state:
$$s_{\textrm{ARB}} = \left(t, d, \tau\right)$$

\begin{itemize}
  \item $t \in \mathbb{R}_{>0}$ is the time when the arbitrage-mint period ends.
  \item $d \in \mathbb{R}_{\ge 0}$ is the remaining amount of dexy stablecoins that can be minted during the current arbitrage-mint period, determined by $t$.
  \item $\tau \in \mathbb{R}_{>0}\cup\{\infty\}$ is the time when $\rr > 1 + \epsilon_{\textrm{FREE}}$. This can occur more than once as time goes on,
        each occurrence corresponding to a potentially new arbitrage-mint period. 
\end{itemize}

The InjectBasecoin action has the following state:
$$s_{\textrm{IB}} = \left(t, \tau\right)$$

\begin{itemize}
  \item $t \in \mathbb{R}_{>0}$ is the time when the previous inject-basecoin action occurred.
  \item $\tau \in \mathbb{R}_{>0}\cup\{\infty\}$ is the time when $\rr < 1 - \epsilon_{\textrm{FREE}}$. This can occur more than once as time goes on,
        each occurrence corresponding to a potentially new inject-basecoin period.
\end{itemize}

The ExtractDexy and ReleaseDexy actions share the following state:
$$s_{\textrm{BURN}} = \left(t_{\textrm{ED}}, t_{\textrm{RD}}, d, \tau^{\textrm{ED}}, \tau^{\textrm{RD}} \right)$$

\begin{itemize}
  \item $t_{\textrm{ED}} \in \mathbb{R}_{>0}$ is the time when the previous extract-dexy action occurred.
  \item $t_{\textrm{RD}} \in \mathbb{R}_{>0}$ is the time when the previous release-dexy action occurred.
  \item $d \in \mathbb{R}_{\ge 0}$ is the remaining dexy amount to be released.
  \item $\tau^{\textrm{ED}} \in \mathbb{R}_{>0}\cup\{\infty\}$ is the time when $\rr \le 1 - \epsilon_{\textrm{ED}}^{0}$. 
  \item $\tau^{\textrm{RD}} \in \mathbb{R}_{>0}\cup\{\infty\}$ is the time when $\rr > 1 + \epsilon_{\textrm{RD}}$.  
\end{itemize}

\subsection{Auxiliary Definitions}

\begin{definition}[Bank Liability]
  The bank liability is defined as the product of 
  the oracle price and the minted supply of dexy stablecoins:
  \begin{equation}
    L(\Pdt, \sbank) = \Pdt \Sd
  \end{equation}
\end{definition}

\begin{definition}[Bank Reserve Ratio]
  The bank reserve ratio is defined as the ratio of its reserves to its liabilities:
  \begin{equation}
    r(\Pdt, \sbank) = \frac{R}{L(\Pdt, \sbank)} = \frac{R}{\Pdt \Sd}
  \end{equation}
  Assuming the oracle price is constant, and the reserves are greater than the liabilities,
  the bank reserve ratio will always increase after stablecoins are minted. 
  During periods of sudden basecoin price crashes, the reserve ratio will decrease. 
  Since users cannot sell dexy stablecoins to the bank, the bank is liable to no one, 
  thus the bank reserves can be less than its liabilities, which is expected during bank interventions.
  So, a measure of the health of the \dx{} protocol is having an increasing reserve ratio, 
  before and after bank interventions.
 
\end{definition}

\begin{definition}[Reference Price]
  The liquidity pool price, or the reference market price, is defined as follows:
  \begin{equation}
    \Plp(\slp) = \frac{B}{D}
  \end{equation}
  The reference price, $\Plp$, has units of $\frac{\textrm{BASECOIN}}{\textrm{STABLECOIN}}$.
\end{definition}

\begin{definition}[Reference Reserve Ratio]
    The ratio of the liquidity pool price to the oracle price is called the reference reserve ratio:
\begin{equation}
  \label{eq:r}
  \bar{r}(\Pdt, \slp) = \frac{\Plp(\slp)}{\Pdt} = \frac{B}{\Pdt D}
\end{equation}
  The aim of the \dx{} protocol is to maintain $\rr$ within a sufficiently small neighbourhood of 1.
  $\rr$ is called the reference reserve ratio since it has a similar structure to the formula for a reserve ratio and 
  we can think of it as the reserve ratio for the liquidity pool.
\end{definition}

\begin{definition}[Arbitrage Mint Maximum]
  In the arbitrage mint state, when $\rr > 1 + \epsilonfree$, the maximum amount of stablecoins the bank 
  will allow to be minted is:
\begin{equation}
  \delta_{d} = \frac{\delta_b}{\tilde{\Pdt}}
\end{equation}
  Where:
\begin{align}
  \delta_b &= B - D * \tilde{\Pdt} \\
  \tilde{\Pdt} &= \Pdt(1 + \phi_{\textrm{BANK}}^{\textrm{FREE}} + \phi_{O})
\end{align}
\end{definition}

\newpage

\subsection{State Update Procedures}

In order to perform each bank action, we must know both the current value of $\rr$ and when
the transition from one region to another occurred. The following state update procedures ensure
that the bank always has the accurate state information available before any bank action is performed.
The procedures are meant to be continuously running, checking which bank action states can be updated
according to whatever value the $\rr$ has. 

In practice, on eUTxO \cite{EUTXO} blockchains, these procedures are validators, such that if they evaluate to $\textsc{true}$,
then the corresponding transaction built to update the state can be successfully submitted to the blockchain for inclusion.

\subsubsection{UpdateArbitrageMint}

\begin{algorithm}
\caption{UpdateArbitrageMint}\label{alg:update-arbmint}
\begin{algorithmic}[1]
\Require $\soracle, \slp, \sarb, \textrm{Update time } t', \textrm{Current time } t$
\State $(\tilde{t}, \tilde{d}, \tilde{\tau}) \gets \sarb$
\State $(\Pdt, F) \gets \soracle$
\State $\rr \gets \rr(\Pdt, \slp)$
\State $notUpdated \gets (\tilde{\tau} == \infty)$
\State $validUpdate \gets (t' \ge t - \tautx) \land (t' \le t)$ \Comment{Update time is in a valid interval.}
\State $notReset \gets (\tilde{\tau} < \infty)$
\State $validReset \gets (t' == \infty)$
\State $isUpdate \gets (\rr > 1  + \epsilonfree) \land notUpdated \land validUpdate$ 
\State $isReset \gets (\rr \le 1 + \epsilonfree) \land notReset \land validReset$
\If{$isUpdate \lor isReset$}
  \State $\sarb \gets (\tilde{t}, \tilde{d}, t')$
  \State \Return \textsc{true}
\Else
  \State \Return \textsc{false}
\EndIf
\end{algorithmic}
\end{algorithm}

The bank requires knowing at what time the $\rr > 1 + \epsilonfree$ for the arbitrage minting period to begin 
or when $\rr \le 1 + \epsilonfree$ to determine the end of the arbitrage minting period. 
Algorithm \ref{alg:update-arbmint} checks that the desired update values provided are valid for either case so the state 
can be updated correctly.

\newpage

\subsubsection{UpdateInjectBasecoins}

\begin{algorithm}[H]
\caption{UpdateInjectBasecoins}\label{alg:update-inject-basecoins}
\begin{algorithmic}[1]
\Require $\soracle, \slp, \sib, \textrm{Update time } t', \textrm{Current time } t$
\State $(\tilde{t}, \tilde{\tau}) \gets \sib$
\State $(\Pdt, F) \gets \soracle$
\State $\rr \gets \rr(\Pdt, \slp)$
\State $notUpdated \gets (\tilde{\tau} == \infty)$
\State $validUpdate \gets (t' \ge t - \tautx) \land (t' \le t)$ \Comment{Update time is in a valid interval.}
\State $notReset \gets (\tilde{\tau} < \infty)$
\State $validReset \gets (t' == \infty)$
\State $isUpdate \gets (\rr < 1  - \epsilonfree) \land notUpdated \land validUpdate$ 
\State $isReset \gets (\rr \ge 1 - \epsilonfree) \land notReset \land validReset$
\If{$isUpdate \lor isReset$}
  \State $\sib \gets (\tilde{t}, t')$
  \State \Return \textsc{true}
\Else
  \State \Return \textsc{false}
\EndIf
\end{algorithmic}
\end{algorithm}

The bank requires knowing at what time the $\rr < 1 - \epsilonfree$ for the inject basecoins period to begin 
or when $\rr \ge 1 - \epsilonfree$ to determine the end of the inject basecoins period. 
Algorithm \ref{alg:update-inject-basecoins} checks that the desired update values provided are valid for either case so the state 
can be updated correctly.

\newpage

\subsubsection{UpdateExtractDexy}

\begin{algorithm}[H]
\caption{UpdateExtractDexy}\label{alg:update-extract-dexy}
\begin{algorithmic}[1]
\Require $\soracle, \slp, \sburn, \textrm{Update time } t', \textrm{Current time } t$
\State $(\tilde{t}_{\textrm{ED}}, \tilde{t}_{\textrm{RD}}, \tilde{d}, \tilde{\tau}^{\textrm{ED}}, \tilde{\tau}^{\textrm{RD}}) \gets \sburn$
\State $(\Pdt, F) \gets \soracle$ 
\State $\rr \gets \rr(\Pdt, \slp)$
\State $notUpdated \gets (\tilde{\tau}^{\textrm{ED}} == \infty)$
\State $validUpdate \gets (t' \ge t - \tautx) \land (t' \le t)$ \Comment{Update time is in a valid interval.}
\State $notReset \gets (\tilde{\tau}^{\textrm{ED}} < \infty)$
\State $validReset \gets (t' == \infty)$
\State $isUpdate \gets (\rr \le 1 - \epsilondexyzero) \land notUpdated \land validUpdate$ 
\State $isReset \gets (\rr > 1 - \epsilondexyzero) \land notReset \land validReset$
\If{$isUpdate \lor isReset$}
  \State $\sburn \gets (\tilde{t}_{\textrm{ED}}, \tilde{t}_{\textrm{RD}}, \tilde{d}, t', \tilde{\tau}^{\textrm{RD}})$
  \State \Return \textsc{true}
\Else
  \State \Return \textsc{false}
\EndIf
\end{algorithmic}
\end{algorithm}

The bank requires knowing at what time the $\rr \le 1 - \epsilondexyzero$ for the extract dexy period to begin 
or when $\rr > 1 - \epsilondexyzero$ to determine the end of the extract dexy period. 
Algorithm \ref{alg:update-extract-dexy} checks that the desired update values provided are valid for either case so the state 
can be updated correctly.

\newpage

\subsubsection{UpdateReleaseDexy}

\begin{algorithm}[H]
\caption{UpdateReleaseDexy}\label{alg:update-release-dexy}
\begin{algorithmic}[1]
\Require $\soracle, \slp, \sburn, \textrm{Update time } t', \textrm{Current time } t$
\State $(\tilde{t}_{\textrm{ED}}, \tilde{t}_{\textrm{RD}}, \tilde{d}, \tilde{\tau}^{\textrm{ED}}, \tilde{\tau}^{\textrm{RD}}) \gets \sburn$
\State $(\Pdt, F) \gets \soracle$ 
\State $\rr \gets \rr(\Pdt, \slp)$
\State $notUpdated \gets (\tilde{\tau}^{\textrm{RD}} == \infty)$
\State $validUpdate \gets (t' \ge t - \tautx) \land (t' \le t)$ \Comment{Update time is in a valid interval.}
\State $notReset \gets (\tilde{\tau}^{\textrm{RD}} < \infty)$
\State $validReset \gets (t' == \infty)$
\State $isUpdate \gets (\rr > 1  + \epsilonreleasedexy) \land notUpdated \land validUpdate$ 
\State $isReset \gets (\rr \le 1 + \epsilonreleasedexy) \land notReset \land validReset$
\If{$isUpdate \lor isReset$}
  \State $\sburn \gets (\tilde{t}_{\textrm{ED}}, \tilde{t}_{\textrm{RD}}, \tilde{d}, \tilde{\tau}^{\textrm{ED}}, t')$
  \State \Return \textsc{true}
\Else
  \State \Return \textsc{false}
\EndIf
\end{algorithmic}
\end{algorithm}

The bank requires knowing at what time the $\rr > 1  + \epsilonreleasedexy$ for the release dexy period to begin 
or when $\rr \le 1 + \epsilonreleasedexy$ to determine the end of the release dexy period. 
Algorithm \ref{alg:update-release-dexy} checks that the desired update values provided are valid for either case so the state 
can be updated correctly.

\newpage

\subsection{Protocol Operations}

\subsubsection{FreeMint}

\begin{algorithm}[H]
\caption{FreeMint}\label{alg:freemint}
\begin{algorithmic}[1]
  \Require $\soracle, \slp, \sbank, \sfree, \textrm{Basecoins } b, \textrm{Updated period end time } T', \textrm{Current time } t$
  \State $(\tilde{t}, \tilde{d}) \gets \sfree$
  \State $(R, \Sd) \gets \sbank$
  \State $(B, D) \gets \slp$
  \State $(\Pdt, F) \gets \soracle$  
  \State $\rr \gets \rr(\Pdt, \slp) $
  \State $\tilde{\Pdt} \gets \Pdt(1 + \phi_{\textrm{BANK}}^{\textrm{FREE}} + \phi_{O})$
  \State $isCounterReset \gets (t > \tilde{t})$ \Comment{If true, then the protocol is in a new FreeMint period.}
  \State $dexyMinted \gets \frac{b}{\tilde{\Pdt}}$
  \State $oracleFeeAmount \gets dexyMinted * \Pdt * \phi_{O}$
  \State $\delta_{d} \gets \alpha_{\textrm{FREE}} * D$ \Comment{The maximum dexy mintable during the FreeMint period.}
  \State $dexyAvailable \gets isCounterReset \textrm{ ? } \delta_{d} \textrm{ : } \tilde{d}$
  \State $validAmount \gets dexyMinted \le dexyAvailable$
  \If{$\lnot isCounterReset$}
    \State $T' \gets \tilde{t}$
    \State $validUpdate \gets \textsc{true}$
  \Else 
    \State $validUpdate \gets T' \in [t + T_{\textrm{FREE}}, t + T_{\textrm{FREE}} + \tautx]$  \Comment{Checks new period end time.}
  \EndIf
  \State $validRatio \gets \rr \in (1-\epsilonfree, 1 + \epsilonfree)$
  \If{$validUpdate \land validRatio$}
    \State $\sbank \gets (R + b, \Sd + dexyMinted)$
    \State $\soracle \gets (\Pdt, F + oracleFeeAmount)$
    \State $\sfree \gets (T', dexyAvailable - dexyMinted)$
    \State \Return \textsc{true}
  \Else
    \State \Return \textsc{false}
  \EndIf
\end{algorithmic}
\end{algorithm}
 
Algorithm \ref{alg:freemint} checks that the protocol is in a valid free-mint range and if so allows the state to change accordingly.

\newpage

\subsubsection{ArbitrageMint}

\begin{algorithm}[H]
\caption{ArbMint}\label{alg:arbmint}
\begin{algorithmic}[1]
  \Require $\soracle, \slp, \sbank, \sarb, \textrm{Basecoins } b, \textrm{Updated period end time } T', \textrm{Current time } t$
  \State $(\tilde{t}, \tilde{d}, \tilde{\tau}) \gets \sarb$
  \State $(R, \Sd) \gets \sbank$
  \State $(B, D) \gets \slp$
  \State $(\Pdt, F) \gets \soracle$  
  \State $\rr \gets \rr(\Pdt, \slp)$
  \State $\tilde{\Pdt} \gets \Pdt(1 + \phi_{\textrm{BANK}}^{\textrm{FREE}} + \phi_{O})$
  \State $isCounterReset \gets (t > \tilde{t})$ \Comment{If true, then the protocol is in a new ArbMint period.}
  \State $dexyMinted \gets \frac{b}{\tilde{\Pdt}}$
  \State $oracleFeeAmount \gets dexyMinted * \Pdt * \phi_{O}$
  \State $\delta_{b} \gets B - D * \tilde{\Pdt}$
  \State $\delta_{d} \gets \frac{\delta_{b}}{\tilde{\Pdt}}$ \Comment{The maximum dexy mintable during the ArbMint period.}
  \State $dexyAvailable \gets isCounterReset \textrm{ ? } \delta_{d} \textrm{ : } \tilde{d}$
  \State $validAmount \gets dexyMinted \le dexyAvailable$
  \If{$\lnot isCounterReset$}
    \State $T' \gets \tilde{t}$
    \State $validUpdate \gets \textsc{true}$
  \Else 
    \State $validUpdate \gets T' \in [t + T_{\textrm{ARB}}, t + T_{\textrm{ARB}} + \tautx]$  \Comment{Checks new period end time.}
  \EndIf
  \State $validDelay \gets t > \tilde{\tau} + \tau_{\textrm{ARB}}$
  \State $validRatio \gets \rr > 1 + \epsilonfree$
  \If{$validUpdate \land validDelay \land validRatio$}
    \State $\sbank \gets (R + b, \Sd + dexyMinted)$
    \State $\soracle \gets (\Pdt, F + oracleFeeAmount)$
    \State $\sarb \gets (T', dexyAvailable - dexyMinted, \tilde{\tau})$
    \State \Return \textsc{true}
  \Else
    \State \Return \textsc{false}
  \EndIf
\end{algorithmic}
\end{algorithm}

Algorithm \ref{alg:arbmint} checks that the protocol is in a valid arbitrage-mint state and if so allows the state to change accordingly.

\newpage 

\subsubsection{InjectBasecoins}

\begin{algorithm}[H]
\caption{InjectBasecoins}\label{alg:inject-basecoins}
\begin{algorithmic}[1]
  \Require $\soracle, \slp, \sbank, \sib, \textrm{Basecoins } b, \textrm{Tx creation time } t', \textrm{Current time } t$
  \State $(\tilde{t}, \tilde{\tau}) \gets \sib$
  \State $(R, \Sd) \gets \sbank$
  \State $(B, D) \gets \slp$
  \State $(\Pdt, F) \gets \soracle$
  \State $\slp' \gets (B + b, D)$
  \State $\rr \gets \rr(\Pdt, \slp) $
  \State $\rr' \gets \rr(\Pdt, \slp')$
  \State $validAmount \gets b \le \alpha_{\textrm{IB}} * R$
  \State $validRatio \gets \rr < 1 - \epsilonfree$
  \State $validInjectionRatio \gets \rr' \in \left[ 1-\epsilonfree, 1 - \epsilonib \right]$
  \State $validDelay \gets t > \tilde{\tau} + \tau_{\textrm{IB}}$
  \State $validInjectionPeriod \gets t > \tilde{t} + T_{\textrm{IB}}$
  \State $validUpdate \gets t \le t' + \tautx$
  \State $validAmounts \gets validAmount \land validRatio \land validInjectionRatio$
  \State $validTimes \gets validDelay \land validInjectionPeriod \land validUpdate$
  \If{$validAmounts \land validTimes$}
    \State $\sbank \gets (R - b, \Sd)$
    \State $\slp \gets (B + b, D)$
    \State $\sib \gets (t', \tilde{\tau})$
    \State \Return \textsc{true}
  \Else
    \State \Return \textsc{false}
  \EndIf
\end{algorithmic}
\end{algorithm}

Algorithm \ref{alg:inject-basecoins} checks that the protocol is in a valid inject-basecoin state and if so allows the state to change accordingly.

\newpage

\subsubsection{ExtractDexy}

\begin{algorithm}[H]
\caption{ExtractDexy}\label{alg:extract-dexy}
\begin{algorithmic}[1]
  \Require $\soracle, \slp, \sburn, \textrm{Dexy } d, \textrm{Tx creation time } t', \textrm{Current time } t$
  \State $(\tilde{t}_{\textrm{ED}}, \tilde{t}_{\textrm{RD}}, \tilde{d}, \tilde{\tau}^{\textrm{ED}}, \tilde{\tau}^{\textrm{RD}}) \gets \sburn$
  \State $(B, D) \gets \slp$
  \State $(\Pdt, F) \gets \soracle$
  \State $\slp' \gets (B, D - d)$ \Comment{Removing Dexy from the liquidity pool.}
  \State $\rr \gets \rr(\Pdt, \slp) $
  \State $\rr' \gets \rr(\Pdt, \slp')$
  \State $validRatio \gets \rr \le 1 - \epsilondexyzero < 1 - \epsilonfree$
  \State $validExtractionRatio \gets \rr' \in \left[ 1-\epsilondexyone, 1 - \epsilonfree \right]$
  \State $validDelay \gets t > \tilde{\tau}^{\textrm{ED}} + \tau_{\textrm{ED}}$
  \State $validBurnPeriod \gets t > \tilde{t}_{\textrm{ED}} + T_{\textrm{BURN}}$
  \State $validUpdate \gets t \le t' + \tautx$
  \State $validAmounts \gets validRatio \land validExtractionRatio$
  \State $validTimes \gets validDelay \land validBurnPeriod \land validUpdate$
  \If{$validAmounts \land validTimes$}
    \State $\slp \gets (B, D-d)$
    \State $\sburn \gets (t', \tilde{t}_{\textrm{RD}}, d, \tilde{\tau}^{\textrm{ED}}, \tilde{\tau}^{\textrm{RD}})$
    \State \Return \textsc{true}
  \Else
    \State \Return \textsc{false}
  \EndIf
\end{algorithmic}
\end{algorithm}

Algorithm \ref{alg:extract-dexy} checks that the protocol is in a valid extract-dexy state and if so allows the state to change accordingly.

\newpage

\subsubsection{ReleaseDexy}

\begin{algorithm}[H]
\caption{ReleaseDexy}\label{alg:release-dexy}
\begin{algorithmic}[1]
  \Require $\soracle, \slp, \sburn, \textrm{Dexy } d, \textrm{Tx creation time } t', \textrm{Current time } t$
  \State $(\tilde{t}_{\textrm{ED}}, \tilde{t}_{\textrm{RD}}, \tilde{d}, \tilde{\tau}^{\textrm{ED}}, \tilde{\tau}^{\textrm{RD}}) \gets \sburn$
  \State $(B, D) \gets \slp$
  \State $(\Pdt, F) \gets \soracle$
  \State $\slp' \gets (B, D + d)$ \Comment{Adding Dexy back into the liquidity pool.}
  \State $\rr \gets \rr(\Pdt, \slp) $
  \State $\rr' \gets \rr(\Pdt, \slp')$
  \State $validRatio \gets \rr \in \left(1 + \epsilonreleasedexy , 1 + \epsilonfree\right)$
  \State $validInjectionRatio \gets \rr' <  1 + \epsilonreleasedexy$
  \State $validDelay \gets t > \tilde{\tau}^{\textrm{RD}} + \tau_{\textrm{RD}}$
  \State $validBurnPeriod \gets t > \tilde{t}_{\textrm{RD}} + T_{\textrm{BURN}}$
  \State $validUpdate \gets t \le t' + \tautx$
  \State $validInjectionAmount \gets d \in (0, \tilde{d})$
  \State $validAmounts \gets validRatio \land validInjectionRatio \land validInjectionAmount$
  \State $validTimes \gets validDelay \land validBurnPeriod \land validUpdate$
  \If{$validAmounts \land validTimes$}
    \State $\slp \gets (B, D+d)$
    \State $\sburn \gets (\tilde{t}_{\textrm{ED}}, t', \tilde{d}-d, \tilde{\tau}^{\textrm{ED}}, \tilde{\tau}^{\textrm{RD}})$
    \State \Return \textsc{true}
  \Else
    \State \Return \textsc{false}
  \EndIf
\end{algorithmic}
\end{algorithm}

Algorithm \ref{alg:release-dexy} checks that the protocol is in a valid release-dexy state and if so allows the state to change accordingly.

\newpage

\section{Theorems}
\label{sec:theorems}

In all theorems, we assume that all algorithms from section \ref{sec:protocol} run correctly so that we need only worry about their 
return values.

The first stability property we consider is the tolerance of the protocol's liquidity pool
to basecoin price crashes.

\begin{theorem}[Basecoin Price Crash Tolerance]
  Let $s$ and $s'$ be the oracle price before and after a basecoin price crash, respectively. 
  Assuming $\Plp$ is constant before and after a basecoin price crash, and 
  that $\rr \in \left[1-\epsilonfree, 1+\epsilonfree \right]$, then the liquidity pool can 
  tolerate a basecoin price crash of at most  $\frac{\rr - (1-\epsilonfree)}{\rr}$ 
  before any intervention.
\end{theorem}

\begin{proof}
  Let $y = \frac{1}{s}$ and $y' = \frac{1}{s'}$ be the basecoin price in units of peg asset before 
  and after the crash, respectively. After the basecoin price crash, the new reference 
  reserve ratio will be:
  
  \begin{align}
    \rr' &= \frac{\Plp}{s'} \\ 
         &= \Plp y' \\
         &\ge 1 - \epsilonfree
  \end{align}

  Since we assume the liquidity pool price remains constant before and after the crash, then:
  
  \begin{equation}
    \Plp = \rr \Pdt = \frac{\rr}{y}
  \end{equation}

  So, we then have that:
  
  \begin{align}
    y' &\ge \frac{1-\epsilonfree}{\Plp} = \frac{(1-\epsilonfree)y}{\rr} \\
       \implies \frac{y - y'}{y} &\le \frac{\rr - (1-\epsilonfree)}{\rr}
  \end{align}
\end{proof}

When there is a crash of the basecoin price, in units of the peg asset as provided by the oracle, 
there will be a time delay before the bank will intervene into the liquidity pool.
This is done to give a chance for the liquidity pool price to stabilize itself such that the bank intervention
will not be necessary. 

In the worst-case scenario, the market, i.e. a subset of $U_{LP}$, fails to correct the liquidity pool price on its own, 
within the time delay, such that $\rr$ is again in the free-mint range. Only then will the bank inject 
part of its basecoin reserve into the liquidity pool to stabilize its price. 

We have the following theorem and corollary quantifying what this minimum basecoin reserve threshold
should be in the worst-case scenario.

\begin{theorem}[Minimum Bank Reserve]
  \label{theorem:minimum-bank-reserve}
  Let $s$ be the oracle price after a basecoin price crash such that $\rr < 1 - \epsilonfree$. 
  Assuming that all users with remaining dexy stablecoins, $u \in Y \subseteq U_{LP}$, collude such that, after a bank intervention, all stablecoins 
  outside of the liquidity pool are sold at the new liquidity pool price $s$, then the bank will need
  a minimum basecoin reserve of $R_{min} \ge sD(\sqrt{\rr} - \rr) + s(\Sd - D)$ to restore the liquidity pool price.
\end{theorem}

\begin{proof}
  Let $B'$ and $D'$ be the balance of basecoins and dexy, respectively, in the liquidity pool after the initial 
  inject basecoin bank intervention. Then, by the CP-AMM formula, we have that $BD = B'D'$. The bank must inject $R_1$ basecoins
  into the liquidity pool such that the new balance will be $B' = B + R_1$. The new liquidity pool price will be $s = \frac{B'}{D'}$, 
  so we then have that $D' = \frac{B + R_1}{s}$. We thus obtain the following relation:
  
  \begin{align}
    BD &= B'D' = (B+R_1)\left(\frac{B+R_1}{s}\right) = \frac{(B+R_1)^2}{s} \\
       &\implies R_1 = \sqrt{sBD} - B
  \end{align}

  Let $D_{OUT}$ be the supply of dexy stablecoins outside the liquidity pool at the time of the intervention.
  Using the bank state $\sbank$, and the liquidity pool state $\slp$, we can calculate the value of this 
  quantity such that $D_{OUT} = \Sd - D$. Thus, to compensate for all $D_{OUT}$ stablecoins being sold
  in the liquidity pool after the intervention, the bank must additionally have an amount of reserves of at least:

  \begin{equation}
    R_2 \ge s D_{OUT} = s(\Sd - D)
  \end{equation}

  So, cumulatively, we now have that: 

  \begin{equation}
    \label{eq:Rmin}
    R_{min} \ge R_1 + R_2 = \sqrt{sBD} - B + s\left(\Sd - D\right)
  \end{equation}

  During an intervention, $\rr = \frac{B}{sD} < 1 - \epsilonfree$, by construction of the protocol.
  So, it follows that $B = \rr sD$. We can then substitute this into equation \ref{eq:Rmin} to obtain the following:

  \begin{align}
    R_{min} &\ge \sqrt{sBD} - B + s\left(\Sd - D\right) \\
            &= \sqrt{\rr s^2 D^2} - \rr sD + s\left(\Sd - D\right) \\
            &= sD \left( \sqrt{\rr} - \rr \right) + s\left(\Sd - D\right)
  \end{align}

\end{proof}

Since it is expected that $R_2$ is very large, it would be unrealistic for the protocol to prevent bank interventions 
before the bank's reserves reached $R_{min}$. The $\rr$ could fall below the free-mint range threshold much sooner than the
bank could adequately fill its reserves from minting fees. 

Furthermore, the collusion assumption in Theorem \ref{theorem:minimum-bank-reserve} is equivalent to there being 
one user or adversary who owns all $D_{OUT}$ stablecoins, which are then sold at the new price $s$. 
If we relax this assumption, then we have the following corollary that provides a practical minimum threshold 
for the bank's basecoin reserves, expressed as a fraction of the liquidity pool dexy stablecoin volume.  

\begin{corollary}[Minimum Liquidity Pool Volume Ratio]
  When $\rr < 1 - \epsilonfree$ and assuming that $D_{OUT}$ is distributed amongst the subset $Y$ of remaining dexy stablecoin holders such that for each $u \in Y$, 
  $\frac{d_u}{D_{OUT}} \ll 1$, then the bank's basecoin reserves should be at least $(\sqrt{\rr} - \rr)$ of the liquidity pool's dexy stablecoin volume
  after the basecoin price crash.
\end{corollary}

\begin{proof}
  From Theorem \ref{theorem:minimum-bank-reserve}, we have that:
  
  \begin{align}
    R_{min} &\ge sD \left( \sqrt{\rr} - \rr \right) + s\left(\Sd - D\right) \\
            &> sD \left( \sqrt{\rr} - \rr \right) \\
  \end{align}
  
  Since $sD$ has units of basecoins, this value represents the volume of dexy stablecoins in the liquidity pool.
  Due to the basecoin price crash, it follows that $s > \Plp \implies sD > B$. Thus, during an inject-basecoin bank intervention, 
  the bank's basecoin reserve, as a fraction of the dexy stablecoin volume in the liquidity pool before the intervention, 
  should be strictly greater than:

  \begin{equation}
    \label{eq:MLPVR}
    \textrm{MLPVR} \equiv \frac{R_{min}}{sD} > \sqrt{\rr} - \rr  
  \end{equation}

\end{proof}

In section \ref{sec:implementations}, implementations of the Dexy protocol set $\epsilonfree = 0.02$.
So, if we suppose a bank intervention occurs at the free-mint range threshold $1 - \epsilonfree = 0.98$, 
then an appropriate value for the $\textrm{MLPVR}$ would be roughly $1 \%$ of the liquidity pool's dexy stablecoin volume.
Thus, the bank intervention would inject basecoins until either $\Plp = s$ or the \textrm{MLPVR} amount from its reserves is used.

Next, we consider the time it takes for the liquidity pool to recover from a price crash via basecoin injection and 
stablecoin extraction. The following theorems show that despite a basecoin price crash, the stability of the protocol is 
ensured through finite recovery convergence times.

\begin{theorem}[Recovery Convergence Time: Basecoin Injection]
  Suppose there is a basecoin price crash greater than $\epsilonfree$ such that $\Pdt = s$. Assuming there is no 
  market correction from users $u \in Y \subseteq \ulp$ and that the bank's basecoin reserve is greater than $R_{min}$, 
  then the liquidity pool price $\Plp$ will recover in a time at most $t \le \tau_{\textrm{IB}} + T_{\textrm{IB}}$.
\end{theorem}

\begin{proof}
  If the basecoin price crashes more than $\epsilonfree$, then UpdateInjectBasecoins will return
  $\textsc{TRUE}$.
  
  Since InjectBasecoins returns two possible values, we consider each case:
  \begin{itemize}
    \item If InjectBasecoins returns $\textsc{TRUE}$, then we know that at least
          $(\tau_{\textrm{IB}} + T_{\textrm{IB}})$ amount of time has elapsed.
    \item If InjectBasecoins returns $\textsc{FALSE}$, then, 
          modulo the non-time based conditions that we assume are $\textsc{TRUE}$, either:
          \begin{itemize}
            \item A time $(\tau_{\textrm{IB}} + T_{\textrm{IB}})$ has not passed yet.
            \item A previous intervention happened before the current one, where $\sib.\tau$ has not changed,
                  so that we must wait at most a time $T_{\textrm{IB}}$ before InjectBasecoins can be called again,
                  since it will be the case that $(t - \sib.t) < T_{\textrm{IB}}$, where $t$ is the current time.
          \end{itemize}
          So, for this case, we will need to wait a time either $\left(T_{\textrm{IB}} + \sib.t - t\right)$ or $\left(\tau_{\textrm{IB}} + T_{\textrm{IB}}\right)$.
  \end{itemize}
  Thus, we will need to wait a time at most $t \le \tau_{\textrm{IB}} + T_{\textrm{IB}}$
  such that the liquidity pool price $\Plp = s$.
\end{proof}

Similarly, if basecoin injection still fails, we want to consider the time it will take for the liquidity pool to recover 
from a price crash via stablecoin extraction.

\begin{theorem}[Recovery Convergence Time: Stablecoin Extraction]
  Suppose there is a basecoin price crash greater than $\epsilondexyzero$ such that $\Pdt = s$. Assuming there is no 
  market correction from users $u \in Y \subseteq \ulp$ and that the bank's basecoin reserve is zero, 
  then the liquidity pool price $\Plp$ will recover in a time at most $t \le \tau_{\textrm{ED}} + T_{\textrm{BURN}}$.
\end{theorem}

\begin{proof}
  If the basecoin price crashes more than $\epsilondexyzero$, then UpdateExtractDexy will return
  $\textsc{TRUE}$.
  
  Since ExtractDexy returns two possible values, we consider each case:
  \begin{itemize}
    \item If ExtractDexy returns $\textsc{TRUE}$, then we know that at least
          $(\tau_{\textrm{ED}} + T_{\textrm{BURN}})$ amount of time has elapsed.
    \item If ExtractDexy returns $\textsc{FALSE}$, then, 
          modulo the non-time based conditions that we assume are $\textsc{TRUE}$, either:
          \begin{itemize}
            \item A time $(\tau_{\textrm{ED}} + T_{\textrm{BURN}})$ has not passed yet.
            \item A previous extraction happened before the current one, where $\sburn.\tau^{\textrm{ED}}$ has not changed,
                  so that we must wait at most a time $T_{\textrm{BURN}}$ before ExtractDexy can be called again,
                  since it will be the case that $(t - \sburn.t_{\textrm{ED}}) < T_{\textrm{BURN}}$, where $t$ is the current time.
          \end{itemize}
          So, for this case, we will need to wait a time either $\left(T_{\textrm{BURN}} + \sburn.t_{\textrm{ED}} - t\right)$ or $\left(\tau_{\textrm{ED}} + T_{\textrm{BURN}}\right)$.
  \end{itemize}
  Thus, we will need to wait a time at most $t \le \tau_{\textrm{ED}} + T_{\textrm{BURN}}$
  such that the liquidity pool price $\Plp = s$.
\end{proof}

So, in the most extreme case, where an inject-basecoin and extract-dexy action occur one after the other and the market fails 
to self-stabilize, the cumulative time for the liquidity pool price to recover will be bounded by the sum of the 
recovery convergence times for each separate action.

\section{Related Works}
\label{sec:related-works}
We first consider the distinction between crypto-backed stablecoins and crypto-collateralized stablecoins.
In crypto-collateralized protocols, only stablecoin holders that provided the initial collateral can redeem their stablecoins for 
a corresponding portion of that collateral. Other stablecoin holders can redeem for the collateral only during 
liquidation events caused by a depeg. This is due to the fact that minted stablecoins function as a loan issued by the bank. 
In crypto-backed protocols, any stablecoin holder always has the right-of-redemption such that they can redeem
their stablecoins for a corresponding portion of the crypto reserves. Equivalently, some fraction of the reserves 
fully backs the stablecoins according to the protocol's reserve ratio.

Both the Djed \cite{DjedEprint} and Gluon \cite{Gluon} protocols are considered as crypto-backed algorithmic stablecoins.
This is because the entire bank reserve is tokenized via two tokens: stablecoins and reservecoins.
Stablecoins represent the liabilities of the bank, while reservecoins represent the reserve surplus, 
i.e. equity, after the liabilities are accounted for. Any stablecoin holder can redeem their stablecoins for some 
portion of the bank reserves. However, the protocol requires that in order to mint stablecoins, 
the value of the reserves provided by a user must be 
greater than the value of stablecoins they receive.

In the Dexy protocol, users can only mint stablecoins; there are no reservecoins that can be minted. 
We can think of the bank as having one reservecoin all to itself. However, even though no user can redeem their 
stablecoins directly from the bank, the entirety of the bank's reserve can be used to restore the liquidity 
pool price in the event of a depeg. These reserves injected into the liquidity pool can then be redeemed by stablecoin holders.
This enables the protocol to have efficient 1-to-1 minting at the oracle price, minus a bank fee. 

In both Djed and Gluon, a depeg is defined as the reserve ratio going below some threshold value set as a protocol parameter. 
For the protocol to return to the peg state, either the oracle price must change, which cannot be controlled, or 
a combination of the amount of bank reserves and the circulating supply of stablecoins must change, which can only happen via user 
actions.

During a depeg, the Dexy protocol employs the two mechanisms of basecoin injection and stablecoin extraction in order 
to stabilize the $\rr$. These actions are done autonomously by the bank. Thus, Dexy exists as an alternative 
crypto-backed stablecoin where recovering from a depeg happens via the bank itself instead of user actions.

In some aspects, \dx{} also resembles Fei \cite{Fei} and Gyroscope \cite{Gyroscope}.
All three are built around an algorithmic central bank: a rule-based protocol entity that
holds reserves, mints stablecoins, and conducts open-market operations in designated markets.
However, \dx{} differs from both in the following ways: 
its bank supports one-way minting only, and this property is immutable, whereas Gyroscope's
primary market allows redemptions and Fei's buy-only bonding curve was later modified by governance.
\dx{}'s intervention rules are likewise immutable, with no governance in principle 
(though this could be included in a particular implementation), and interventions can only
take place after predefined time delays, giving markets time to self-stabilize.
Finally, \dx{}'s last-resort extraction of stablecoins is reversible, unlike Fei's permanent burns.

\section{Implementations}
\label{sec:implementations}

In this section we consider two concrete implementations of the \dx{} protocol that were made for the Ergo blockchain, 
a Gold-pegged stablecoin named \dxg{} and a USD-pegged stablecoin named \dxusd{}. The oracles used in both stablecoins are
implemented using the same underlying oracle protocol.

\subsection{Oracle Pools}
\label{subsec:oracle}

In section \ref{sec:protocol}, the Dexy protocol assumes a trusted oracle, independent of implementation. 
In practice, in order for the assumption of a trusted and reliable oracle to be satisfied, the oracle should 
have an incentivization mechanism. Here, we briefly explain the incentivization mechanism of the oracle 
protocol used for \dxg{} and \dxusd{}.

The Oracle Pools 2.0 framework~\cite{eip23} has support for paying rewards in custom tokens. For each instance of an oracle, 
representing a different peg asset, a corresponding amount of ORTs (Oracle Reward Tokens) will be created. 
These ORTs will be used to reward oracles and developers, e.g. for bounties or audits.

In the Dexy protocol, only the bank minting transactions charge an oracle fee in basecoins. 
These basecoins are sent to a BuyBack contract, which purchases the corresponding amount of ORT available in a liquidity pool on a DEX. 
The purchased ORT tokens are then distributed to the oracle operators. Thus, the oracle operators receive ORT tokens that they 
can then redeem for the accumulated basecoin fees in the same DEX liquidity pool.

\subsection{Governance}
\label{subsec:governance}

Both Dexy implementations use the same governance mechanism that allows for updating either the bank contract, 
the basecoin injection contract, or the stablecoin extraction and release contract. This is done via ballot tokens used 
for casting votes on whether or not to update a particular contract to a new contract. Updates are no longer possible if either all 
of these ballot tokens are burned, thus no one can vote, or if the update contract removes the spending path for future updates.
See the link in subsection \ref{subsec:smartcontracts} for implementation details.

\subsection{\dxg}
\label{subsec:dexygold}

The first variant of the Dexy protocol is pegged to Gold and has been implemented on the Ergo blockchain. The current protocol 
specification outlined in Section \ref{sec:protocol} is modeled after this implementation.

The ORT token for the Gold oracle pool is called GORT. The pool reports the price for 1 kilogram of Gold, however, 
the stablecoin protocol implementation uses 1 milligram of Gold internally as the peg.

A UI for interacting with the \dxg{} protocol is available at \url{https://cruxfinance.io/dexy/mint}, 
when selecting the DexyGold option.

\subsection{\dxusd}
\label{subsec:dexyusd}

A variant of the Dexy protocol pegged to USD has also been implemented on the Ergo blockchain, it is marketed using the name USE.

In Section \ref{sec:protocol}, the protocol specification defined $\alpha_{IB}$ as the percentage of the bank's basecoin reserve
that determines the max amount of basecoins to inject into the liquidity pool. USE differs from this specification 
by assigning $\alpha_{IB}$ as the percentage of the liquidity pool instead of the bank reserve. 
This ensures that the bank reserves are not depleted as fast, since it is expected that its basecoin reserves exceed 
the basecoin liquidity in the liquidity pool.

A UI for viewing stats of the \dxusd{} protocol is available at \url{https://cruxfinance.io/dexy/analytics}, 
while interacting with the algorithmic central bank can be done at \url{https://cruxfinance.io/dexy/mint}, 
when selecting the USE option.

\subsection{Smart-Contracts}
\label{subsec:smartcontracts}
The ErgoScript smart-contracts follow a separation-of-concerns design pattern. Each one handles its own part of the protocol, 
instead of having one giant contract. This results in a total of 18 ErgoScript contracts required for the implementation. 
These 18 contracts are orchestrated together with the corresponding off-chain software for the different bank actions and
the trackers necessary to check the status of the oracle and pool peg, and the relevant delay periods.
The smart-contracts and off-chain code can be found in the following repository \url{https://github.com/kushti/dexy-stable}. 

\bibliographystyle{plain}

\bibliography{sources}

\end{document}